\documentclass[letterpaper, 10 pt, journal]{IEEEtran}
\usepackage{cite}
\usepackage{xcolor}
\IEEEoverridecommandlockouts                              

\usepackage{enumitem}

\usepackage{algorithm, algpseudocode, amssymb, amsthm,amsmath,amsfonts,verbatim, mathtools,dsfont, bbold}
\usepackage{cite}
\usepackage{diagbox}
\usepackage{times,bm,tikz}
\usepackage{booktabs}
\usepackage[hidelinks]{hyperref}
\allowdisplaybreaks[2]
\newcommand{\ones}{\mathbf 1}
\newcommand{\reals}{{\mathbb{R}}}

\newcommand{\argmin}{\mathop{\rm argmin}}

\newcommand{\norm}[1]{\left\lVert#1\right\rVert}
\newcommand{\mnorm}[1]{{\left\vert\kern-0.25ex\left\vert\kern-0.25ex\left\vert #1 
    \right\vert\kern-0.25ex\right\vert\kern-0.25ex\right\vert}}

\newcommand{\mc}{\mathcal}

\newtheorem{definition}{Definition} 
\newtheorem{theorem}{Theorem}

\newtheorem{corollary}{Corollary}
\newtheorem{remark}{Remark}
\newtheorem{proposition}{Proposition}
\newtheorem{assumption}{Assumption}
\newtheorem{example}{Example}

\algnewcommand\algorithmicforeach{\textbf{for each}}
\algdef{S}[FOR]{ForEach}[1]{\algorithmicforeach\ #1\ \algorithmicdo}

\renewcommand*\Call[2]{\textproc{#1}(#2)}

\newcommand{\change}[1]{#1} 

\newcommand\copyrighttext{%
  \footnotesize Copyright (c) 2022 IEEE. Personal use is permitted. For any other purposes, permission must be obtained from the IEEE by emailing pubs-permissions@ieee.org2022 IEEE. 
This is the author's version of an article that has been published in this journal. Changes were made to this version by the publisher prior to publication.
The final version of record is available at \href{http://dx.doi.org/10.1109/LCSYS.2020.300124}{10.1109/LCSYS.2020.3001240}}
\newcommand\copyrightnotice{%
\begin{tikzpicture}[remember picture,overlay]
\node[anchor=south,yshift=5pt] at (current page.south) {\fbox{\parbox{\dimexpr\textwidth-\fboxsep-\fboxrule\relax}{\copyrighttext}}};
\end{tikzpicture}%
}

\begin{document}
\pagestyle{empty} 
\title{\Large \bf Congestion-aware path coordination game with Markov decision process dynamics}
\author{Sarah H. Q. Li$^{1}$, Daniel Calderone$^{1}$, Beh\c cet\ A\c c\i kme\c se$^{1}$
\thanks{$^{1}$Authors are with the William E. Boeing  Department of Aeronautics and Astronautics, University of Washington, Seattle. 
        {\tt\small sarahli@uw.edu}
        {\tt\small djcal@uw.edu }
        {\tt\small behcet@uw.edu}}%
}

\maketitle
\copyrightnotice
\thispagestyle{empty} 
\begin{abstract}
Inspired by the path coordination problem arising from robo-taxis, warehouse management, and mixed-vehicle routing problems, we model a group of heterogeneous players responding to stochastic demands as a congestion game under Markov decision process dynamics.
Players share a common state-action space but have unique transition dynamics, and each player's unique cost is a \emph{function} of the joint state-action probability distribution. 
 \change{For a class of player cost functions, we formulate the player-specific optimization problem, prove equivalence between the Nash equilibrium and the solution of a potential minimization problem, and derive dynamic programming approaches to solve for the Nash equilibrium. We apply this game to model multi-agent path coordination and introduce congestion-based cost functions that enable players to complete individual tasks while avoiding congestion with their opponents. }
 Finally, we present a learning algorithm for finding the Nash equilibrium that has linear complexity in the number of players. We demonstrate our game model on a multi-robot warehouse \change{path coordination problem}, in which robots autonomously retrieve and deliver packages while avoiding congested paths. 
\end{abstract}
\section{Introduction}\label{sec:introduction}
As autonomous \change{path} planning algorithms become widely-adapted by aeronautical, robotics, and operational sectors~\cite{yun2020multi,ota2006multi}, the standard underlying assumption that the operating environment is stationary is no longer sufficient. More likely, autonomous players \emph{share} the operating environment with other players who may have conflicting objectives. While the possibility for multi-agent conflicts has pushed single-agent \change{path} planning towards greater emphasis on \change{robust planning} and collision avoidance, we believe that the overarching goal should be to consider other players' trajectories and achieve optimality with respect to the \emph{multi-agent dynamics}. 

We focus on the scenario where a group of heterogeneous players collectively perform path planning in response to stochastic demands. We are inspired by fleets of robo-taxis fulfilling ride demands while avoiding congestion in traffic~\cite{vosooghi2019robo} and warehouse robots retrieving packages under dynamic arrival rates~\cite{kumar2018development,li2020mechanism} while avoiding collisions. The common feature in these applications is that the players must plan with respect to a forecasted demand distribution rather than a deterministic demand. We assume that the desirable outcome is a competitive equilibrium. Beyond competitive settings, a competitive equilibrium can be used in cooperative settings to ensure that each player achieves identical costs and each demand is \emph{optimally} fulfilled with respect to other demands, thus ensuring a degree of fairness.

We propose MDP congestion games as \change{a theoretical framework for analyzing the resulting path coordination problem}. By leveraging common congestion features in multi-agent \change{path} planning, our key contribution is \emph{reducing} the $N$-player coupled MDP problem to a \emph{single} potential minimization problem. As a result, we can use optimization techniques to analyze the Nash equilibrium as well as apply gradient descent methods to compute it.



\textbf{Contributions}. To address the lack of \change{game-theoretical models for} path \change{coordination under} MDP dynamics, we propose an MDP congestion game with finite players and heterogeneous \change{player} costs and dynamics. We define Bellman equation-type conditions for the Nash equilibrium, formulate a \emph{potential function} and provide a necessary and sufficient condition for its existence. \change{Under certain assumptions on the player costs}, we show equivalence between the \change{Nash equilibrium} and the \change{global solution of} the potential minimization problem, and provide sufficient conditions for a unique Nash equilibrium. Specifically \change{for multi-player path coordination}, we \change{formulate a class of cost functions that allows players to have different sensitivities to the total congestion and to find congestion-free paths that optimally achieve their individual objectives}. Finally, we provide a \change{distributed} algorithm that converges to the Nash equilibrium and give rates of \change{its} convergence. We demonstrate our model and algorithm on a 2D autonomous warehouse problem where robots retrieve and deliver packages \change{with stochastic arrival times} while sharing a common navigation space.

\section{Related work}
\label{sec:related work}

An MDP congestion game~\cite{calderone2017infinite} is a stochastic population game and is related to potential mean field games \cite{lasry2007mean,gueant2011infinity} in the discrete time and state-action space \cite{gomes2010discrete} and mean field games on graphs \cite{gueant2015existence}. 
\change{In this paper, we extend our previous framework from continuous populations of identical MDP decision makers~\cite{calderone2017infinite} to a finite number of heterogenous MDP decision makers.} In the continuous population case, MDP congestion games have been analyzed for constraint satisfaction in~\cite{li2019tolling} and sensitivity to \change{hyper}parameters in~\cite{li2019sensitivity}.  

Model-based multi-agent \change{path} planning is typically solved via graph-based searches~\cite{cohen2019optimal} and mixed integer linear programming~\cite{chen2020scalable}. Recently, reinforcement learning has been introduced as a viable method for solving multi-agent \change{path} planning~\cite{semnani2020multi,yun2020multi}. In most scenarios, the \change{path} planning problem is modeled as an MDP~\cite{bayerlein2021multi,lo2021towards}. In particular, \cite{lo2021towards} adopts a stochastic game model for human-robot collision avoidance, but focuses more on algorithm development rather than game structure analysis. 

\section{Heterogeneous MDP Congestion Game}\label{sec:notation}
Consider a finite number of players $[N] = \{1, \ldots, N\}$ with a \emph{shared finite state-action space} given by $([S], [A])$ and common time interval $\mc{T} = \{0, 1, \ldots, T\}$. Each player $i$ has \emph{individual} time-dependent transition probabilities given by $P^i\in \reals_+^{TSSA}$, where at time $t$, $P^i_{ts'sa}$ is the transition probability from state $s$ to state $s'$ using action $a$ \change{satisfying} the simplex constraints:
\begin{equation}
    \textstyle\sum_{s'} P^i_{t s' s a} = 1, \ \forall (i, t, s,  a) \in [N]\times \change{[T]}\times [S] \times [A].
\end{equation}
\textbf{State-action distribution}. \change{At time $t$, let player $i$'s state be $s^i(t) \in [S]$ and action taken be $a^i(t) \in [A]$, then $x^i_{tsa} = \mathbb{P}[s^i(t) = s, a^i(t) = a]$ is  player $i$'s probability of being in state $s$ taking action $a$ at time $t$.}
\change{Player $i$'s state-action probability trajectory over time period $\mc{T}$ is $x^i \in\reals^{(T+1)SA}$, its \emph{state-action distribution}.} \change{We use $\mc{X}(P^i, z^i_0)$ to denote the set of all feasible state-action distributions under transition dynamics $P^i$ and initial condition $z^i_0 \in\reals^S_+$, where $z^i_{0s}= \mathbb{P}[s^i(0) = s]$ is player $i$'s probability of starting in state $s$.}
\begin{multline}\label{eqn:feasible_mdp_flows}
\textstyle\mc{X}(P^i, \change{z^i_{0}}):= \Big\{\change{x^i} \in \reals_+^{(T+1)SA} \Bigg| \sum_{a} \change{x^i_{0sa} = z^i_{0s}}, \forall s \in [S], \Big. \\ 
\textstyle\Big.\sum_{s', a} P^i_{tss'a} \change{x^i_{(t-1)s'a}} = \sum_{a} \change{x^i_{tsa}}, \ \forall (t,s) \in  \change{[T]}\times [S] \Big\}.
\end{multline}
The \emph{joint state-action distribution} of all players is given by
\begin{equation}\label{eqn:joint_state_action}
    x = (x^1, \ldots, x^N) \in \reals_+^{N(T+1)SA}.
\end{equation}
\change{We assume that $x$ is fully observable and may denote it as $x = (x^i, x^{-i})$ where $x^{-i} = (x^j)_{j \in [N]/\{i\}}$.}

\noindent\textbf{\change{Player} costs}. \change{Similar to stochastic games, the player} costs are continuously differentiable \emph{functions} of $x$: player $i$ incurs a cost $\ell^i_{tsa}(x)$ for taking action $a$ at state $s$ and time $t$.
\begin{equation}\label{eqn:player_cost_def}
    \ell^i_{tsa}: \reals_+^{N(T+1)SA} \mapsto \reals, \ \forall (i,t,s,a) \in [N]\times \mc{T} \times [S] \times [A].
\end{equation}
\change{Compared to stochastic games where player costs are coupled to the opponent policies,~\eqref{eqn:player_cost_def} is better suited to model collision events. For example, the expectation of the log-barrier function for players $i$ and $j$ at time $t$ can be modeled as $\sum_{s, s' \in [S]} (\sum_a x^i_{tsa})(\sum_a x^j_{ts'a})\log(d_{s,s'})$, in which $d_{s,s'}$ denotes the distance between states $s, s'\in [S]$. }

\noindent The \emph{cost vector} of $(\ell^1,\ldots \ell^N)$~\eqref{eqn:player_cost_def} is given by $\xi:\reals_+^{N(T+1)SA}\mapsto \reals_+^{N(T+1)SA}$,
\begin{equation}\label{eqn:cost_vector}
    \xi(x) = [\ell^1_{011}(x), \ell^1_{012}(x), \ldots, \ell^N_{TSA}(x)] \in \reals_+^{N(T+1)SA}.
\end{equation}
We assume that $\xi$ has a positive definite gradient in $x$. 
\begin{assumption}\label{assum:positive_definite} The player cost vector $\xi$~\eqref{eqn:cost_vector} satisfies $\nabla \xi(x) \succ 0$ for all $x$~\eqref{eqn:joint_state_action} where $x^i \in \mc{X}(P^i, x_0^i), \ \forall i \in [N]$.
\end{assumption}
\noindent For the class of player costs considered in Section~\ref{sec:path_coordination_costs}, Assumption~\ref{assum:positive_definite} implies that the player costs strictly increase as the number of players increases.

\noindent\textbf{Coupled MDPs}. Given an initial distribution $z_0^i \in \reals^S_+$ and fixed state-action distributions \change{$x^{-i}$}~\eqref{eqn:joint_state_action}, player $i$ solves the following optimization problem under MDP dynamics.
\begin{equation}
\begin{aligned}\label{eqn:individual_player_mdp}
     \underset{x^i}{\mbox{min}} & \change{\sum_{t, s, a}\int_0^{x^i_{tsa}}\ell^i_{tsa}(u^i, x^{-i}) \partial u^i_{tsa}} \ \mbox{s.t. } & x^i \in \mc{X}(P^i, z^i_{0}).
\end{aligned}
\end{equation}
\change{In~\eqref{eqn:individual_player_mdp}, we note that each integral is taken over $u^i_{tsa}$, the $(t, s, a)^{th}$ element of $u^i$.} 
When $\ell^i_{tsa}(x)$ is constant for all $(t,s,a) \in \mc{T} \times [S]\times [A]$, player $i$ solves a standard \emph{linear program} MDP.

\noindent\textbf{Dynamic programming}. At a joint state-action distribution $x$~\eqref{eqn:joint_state_action}, player $i$'s cost-to-go in~\eqref{eqn:individual_player_mdp} can be recursively defined via Q-value functions~\cite{puterman2014markov} as
\begin{multline}
   Q^i_{Tsa}(x) := \ell^i_{Tsa}(x), \label{eqn:q_value} \\ 
   Q^i_{(t-1)sa}(x) := \ell^i_{(t-1)sa}(x) + \textstyle\sum_{s'} P^i_{ts'sa}\underset{a'}{\min}\, Q^i_{t,s'a'}(x),\\
\forall \ t \in [T] 
\end{multline}
\noindent\change{
The optimal solution of~\eqref{eqn:individual_player_mdp} can be stated using~\eqref{eqn:q_value}.}

\begin{theorem}
\change{Under Assumption~\ref{assum:positive_definite}, ${x}^i$~\eqref{eqn:feasible_mdp_flows} uniquely minimizes~\eqref{eqn:individual_player_mdp} with respect to the state-action distribution $x^{-i}$ if and only if its associated $Q^i({x}^i, x^{-i})$~\eqref{eqn:q_value} satisfies
\begin{equation}\label{eqn:individual_optimality}
    {x}^i_{tsa} > 0 \Rightarrow \textstyle Q^i_{tsa}({x}^i, x^{-i})= \min_{a'} Q^i_{tsa'}({x}^i, x^{-i}),
\end{equation}
for all $(t, s, a) \in \mc{T}\times[S]\times[A]$.} I.e., ${x}^i$ is optimal for~\eqref{eqn:individual_player_mdp} if and only if every action played with nonzero probability achieves the minimum cost-to-go~\eqref{eqn:q_value} among available actions. 
\end{theorem}
\begin{proof}
\change{
Let $F({x}^i, x^{-i}) = \sum_{t, s, a}\int_0^{{x}^i_{tsa}}\ell^i_{tsa}(u^i, x^{-i}) \partial u^i_{tsa}$, then $\partial F({x}^i, x^{-i}) /\partial x^i = \ell({x}^i, x^{-i})$. We then apply Proposition A\ref{prop:dp_to_minimizer} to~\eqref{eqn:individual_player_mdp} and the theorem's results follow directly.}
\end{proof}
\noindent\change{
When all players jointly achieve the optimal cost-to-go~\eqref{eqn:individual_optimality}, a stable equilibrium for unilateral optimality is achieved.
}
\begin{definition}[Nash Equilibrium] \label{def:wardrop} 
The joint state-action distribution $\hat{x}= [\hat{x}^1, \ldots, \hat{x}^N]$~\eqref{eqn:joint_state_action}
is a \emph{Nash equilibrium} if \change{$\big(\hat{x}^i, Q^i(\hat{x})\big)$} satisfies~\eqref{eqn:individual_optimality} for all $i \in [N]$.
\end{definition}
\subsection{Potential optimization form}
\noindent We are interested in MDP congestion games that can be reduced from the coupled MDPs~\eqref{eqn:individual_player_mdp} to a single minimization problem given by
\begin{equation}
    \begin{aligned}\label{eqn:convex_opt_eqn}
        \min_{x^1, \ldots, x^N}  F(x), \ 
    \text{s.t. } x^i \in \mc{X}(P^i, z^i_0) , \ \forall \ i \in [N], 
    \end{aligned}
\end{equation}
where $F$ is the \emph{potential function} of the corresponding game.

\begin{definition}[Potential Function]\label{def:potential}
We say an MDP congestion game with {player costs} $\{\ell^{i}\}_{i \in [N]}$~\eqref{eqn:player_cost_def} has a \emph{potential function} $F: \reals^{N(T+1) S A}\mapsto \reals$ if $F$ satisfies
\begin{equation}\label{eqn:potential_first_order}
        \frac{\partial F(x)}{\partial x^i_{tsa}} = \ell^i_{tsa}(x), \ \forall \ (i, t, s, a) \in [N]\times\mc{T}\times[S]\times [A].
    \end{equation}
\end{definition}
\noindent The following assumption on $\{\ell^{i}\}_{i \in [N]}$ is necessary and sufficient for the existence of $F$~\cite[Eqn.2.44]{patriksson2015traffic}.
\begin{assumption}\label{ass:potential_existence}
For all $(i, t, s, a) , (i', t', s', a') \in [N]\times\mc{T}\times[S]\times [A]$, the player costs $\{\ell^{i}\}_{i \in [N]}$ satisfy
    \begin{equation}\label{eqn:potential_second_order}
        \frac{\partial \ell^i_{tsa}(x)}{\partial x^{i'}_{t's'a'}} = \frac{\partial \ell^{i'}_{t's'a'}(x)}{\partial x^i_{tsa}}.
    \end{equation}
\end{assumption}
\begin{remark}
Assumption~\ref{ass:potential_existence} \change{is equivalent to} $F$ \change{being} conservative: $\forall \ x_1, x_2 \in \{x^i_{tsa}\ | \ (i, t, s, a) \in [N] \times \mc{T} \times [S] \times [A]\}$,   
\begin{equation}\label{eqn:potential_conservative}
    {\partial^2 F(x)}/{\partial x_1\partial x_2} = {\partial^2 F(x)}/{\partial x_2\partial x_1}.
\end{equation}
In other words, the Jacobian of $\xi$~\eqref{eqn:cost_vector}, ${\partial\xi(x)}/{\partial x}$, is symmetrical.
\end{remark}
\noindent Verifying the existence of $F$~\eqref{eqn:potential_first_order} is non-trivial. However, if $F$ exists, the \change{solution} of~\eqref{eqn:convex_opt_eqn} is the Nash equilibrium~\cite{calderone_finite}.  
\begin{theorem}\label{thm:kkt_pts} 
\change{If the player costs $\{\ell^{i}\}_{i \in [N]}$~\eqref{eqn:player_cost_def} satisfy Assumption~\ref{assum:positive_definite},  
\begin{enumerate}
    \item the potential function (Definition~\ref{def:potential}) exists,
    \item $\hat{x}$~\eqref{eqn:joint_state_action} is the global optimal solution of~\eqref{eqn:convex_opt_eqn} if and only if $\hat{x}$ is a Nash equilibrium (Definition~\ref{def:wardrop}). 
\end{enumerate}
}
\end{theorem}
\begin{proof}
\change{We prove statement $1$ by showing that Assumption~\ref{assum:positive_definite} implies Assumption~\ref{ass:potential_existence}: if $\nabla \xi(x)\succ 0$ for all feasible joint state-action distributions $x$~\eqref{eqn:joint_state_action}, then $\nabla \xi(x)$ is symmetrical and satisfies~\eqref{eqn:potential_second_order}.}
\change{Next, we show the forward direction of the statement $2$. If $(\hat{x}^1,\ldots \hat{x}^N)$ minimizes~\eqref{eqn:convex_opt_eqn}, then for each $i\in [N]$, $\hat{x}^i$ minimizes~\eqref{eqn:mdp_variant} at $\hat{x}^{-i}$. From Proposition A\ref{prop:dp_to_minimizer}, $\hat{x}^i$ satisfies~\eqref{eqn:individual_optimality} for all $i\in [N]$, therefore $\hat{x}$ is a Nash equilibrium. To show the reverse direction of $2$, if~\eqref{eqn:individual_optimality} is satisfied for all $i \in [N]$, $\hat{x}^i$ is coordinate-wise optimal for coordinate $i$ (Proposition A\ref{prop:dp_to_minimizer}). Under Assumption~\ref{assum:positive_definite},~\eqref{eqn:convex_opt_eqn} has a strictly convex differentiable objective with separable convex constraints $\mc{X}(P^i, z^i_0)$---each $x^i$ is constrained independently of $x^j$, $\forall j \in[N]/\{i\}$, then the jointly coordinate-wise optimal $\hat{x}$ is the global optimal solution of~\eqref{eqn:convex_opt_eqn}~\cite[Thm 4.1]{tseng2001convergence}. }
\end{proof}
\subsection{Path Coordination as an MDP Congestion Game}\label{sec:path_coordination_costs}
\noindent We now model the path coordination problem as an MDP congestion game and demonstrate how players can achieve individual objectives while avoiding each other.

\noindent To reflect the congestion level of each state-action, we first define a \textbf{congestion distribution} as the weighted sum of individual state-action distributions.
\begin{equation}\label{eqn:congestion_distribution}
    \textstyle y := \sum_{i\in[N]} \alpha_ix^i \in \reals^{(T+1)SA}, \ \alpha_i > 0, \  \forall i \in [N],
\end{equation}
\noindent where $\alpha_i$ is player $i$'s \emph{impact factor}. If all players contribute to congestion equally, $\alpha_i  = 1 \ \forall i \in [N]$.  

\noindent\textbf{Player costs}. \change{We derive a class of player costs that satisfy Assumption~\ref{assum:positive_definite}, incorporate congestion-based penalties, and enable players to pursue individual objectives. For all $(i,t,s,a) \in [N]\times\mc{T}\times[S]\times[A]$, the player cost is given by }
\begin{equation}\label{eqn:individual_costs}
  \textstyle\ell^i_{tsa}(y, x^i) = \alpha_i f_{ts}\big(\sum_{a'} y_{tsa'}\big) + \alpha_i g_{tsa}\big(y_{tsa}\big) + h^i_{tsa}(x^i_{tsa}),
\end{equation}
where \change{$\alpha_i$ is the same as in~\eqref{eqn:congestion_distribution}, $f_{ts}:\reals \mapsto\reals$ is the state-dependent congestion and takes the congestion level of $(t,s)$ as input, $g_{tsa}:\reals \mapsto\reals$ is the state-action-dependent congestion and takes the congestion level of $(t,s,a)$ as input, and $h^i_{tsa}:\reals\mapsto\reals$ is the player-specific objective and takes player $i$'s probability of being in $(t,s,a)$ as input. Player-specific objectives such as obstacle avoidance and target reachability can be incorporated as constant offsets in $h^i$}. 
\begin{remark}[Effect of $\alpha_i$]
The impact factor $\alpha_i$ scales player $i$'s \change{relative impact on the total congestion and the total congestion's impact on player $i$. When $\alpha_i < \alpha_j$, player $i$ impacts congestion less and cares about the congestion less than player $j$. When $\alpha_i > \alpha_j$, player $i$ impacts congestion more and cares about the congestion more than player $j$. }
\end{remark}
\noindent The potential function~\eqref{eqn:potential_first_order} of the game with costs~\eqref{eqn:individual_costs} is
\begin{equation}
 \begin{aligned}\label{eqn:state_state_action_potential}
    F(x)= & 
\textstyle   \sum_{t,s} \int_0^{\sum_{a'}y_{tsa'}}f_{ts}(u) \partial u + \sum_{t,s,a} \int_0^{y_{tsa}}g_{tsa}(u) \partial u   \\
    & \textstyle 
    + \sum_{i, t, s, a} \int_{0}^{x^i_{tsa}} h^i_{tsa} (u)\partial u.
\end{aligned}   
\end{equation}
\begin{remark}
\change{Congestion costs $f$ and $g$ must be identical for all players in order for a potential (Definition~\ref{def:potential}) to exist.}
\end{remark}
\begin{example}[Road-sharing Vehicles]
\change{
Consider a sedan (player $1$, $\alpha_1=1$) and a trailer (player $2$, $\alpha_2=2$) sharing a road network modeled by $[S]\times [A]$. Player $i$ wants to reach state $s_i \in [S]$. The player-specific objective is $h^i_{tsa}(x^i_{tsa}) = -\mathbb{1}[s = s_i] + \epsilon_i x^i_{tsa}$, where $\mathbb{1}[w]$ is $1$ when $w$ is true and $0$ otherwise. The term $\epsilon_i x^i_{tsa}$ where $\epsilon_i > 0$ encourages player $i$ to randomize its policy over all optimal actions. Players experience state-based congestion as $f_{ts}(w) =\exp(w)$. The player cost~\eqref{eqn:individual_costs} is $\ell^i_{tsa}(y, x^i) = \alpha_i\exp(\sum_{a'} y_{tsa'}) + \epsilon_ix^i_{tsa} -\mathbb{1}[s = s_i]$.  }
\end{example}

\begin{corollary}\label{cor:unique_ne}
Player costs of form~\eqref{eqn:individual_costs} \change{satisfy Assumption~\ref{assum:positive_definite} if  $h^i_{tsa}(\cdot)$ is strictly increasing and $f_{ts}(\cdot)$, $g_{tsa}(\cdot)$ are non-decreasing} $\forall (i,t,s,a)\in [N]\times\mc{T}\times[S]\times[A]$.
\end{corollary}
\begin{proof}
 \change{Let $I_Z$ be an identity matrix of size $Z\times Z$, $\ones_Z$ be a ones vector of size $Z\times 1$, $\vec{\alpha} = [\alpha_1,\ldots,\alpha_N] \in \reals^{N\times 1}$, $h(x) = [h^1(x),\ldots,h^N(x)] \in \reals^{N(T+1)SA}$, and $\otimes$ be a kronecker product. We define the matrices $M = \vec{\alpha} \otimes I_{(T+1)SA}$ and $J = (I_{(T+1)S}\otimes \ones_A^\top)M$, and verify that $Mx = y$, $[Jx]_{ts} = \sum_{a'}y_{tsa'}$ $\forall (t,s) \in \mc{T}\times[S]$, and $\xi(x) = J^\top f(Jx) + M^\top g(Mx) + h(x)$. Let $w = Jx$, we can take $\xi$'s gradient as
$
\textstyle \nabla \xi(x) = 
 J^\top \nabla f(w) J + 
 M^\top \nabla g(y) M + \nabla h(x).
 $
Under Corollary assumptions,
$\nabla f(w)$ and $\nabla g(y)$ are non-negative diagonal matrices and $\nabla h(x) $ is a strictly positive diagonal matrix. Therefore}, $\nabla \xi(x) \succ 0$. 
\end{proof}
\begin{remark}
Corollary~\ref{cor:unique_ne} implies that a strictly increasing $h^i$ is crucial to ensuring a unique Nash equilibrium. 
Therefore, $h^i$ can be interpreted as a regularization term.
\end{remark}
\subsection{Frank-Wolfe Learning Dynamics}
\noindent We find the Nash equilibrium of MDP congestion games by {leveraging} single-agent dynamic programming. 
\begin{algorithm}[ht!]
\caption{Frank-Wolfe with dynamic programming}
\begin{algorithmic}[1]
\Require \(\{\ell^i\}_{i\in[N]}\), \(\{P^i\}_{i\in[N]}\), \(\{z^i_0\}_{i\in[N]}\), \(N\), \([S], [A], \mc{T}\).
\Ensure \(\change{\{\hat{x}^{i}_{tsa}\}_{t \in \mc{T}, s\in[S], a \in [A]}}\).
\State{\( x^{i0}\in\mc{X}(P^i, z^i_0) \in \reals^{(T+1)SA}, \quad \forall \ i \in [N]. \)}

\For{\(k = 1, 2, \ldots, \)}
\For{\(i = 1,\ldots,N\)}
    \State{\(\change{C^{ik}} = \ell^i([x^{1k},\ldots, x^{Nk}])\)}\label{alg:cost_retrieval}
	\State{\(\pi^i\) = \text{MDP}(\change{\(C^{ik}\)}, \(P^i\), \([S]\), \([A]\), \(T\), \change{\(z^i_0\)})}\label{alg:mdp}
	\State{\(b^{ik} = \) \Call{RetrieveDensity}{\(P\), \(z^i_0\), \(\pi^i\)}}\Comment{Alg.~\ref{alg:density}}
	\State{\(x^{i(k+1)}= (1 -  \frac{2}{k+1})x^{ik} +  \frac{2}{k+1} b^{ik}\)}
\EndFor
\EndFor
\end{algorithmic}
\label{alg:frank_wolfe}
\end{algorithm}

\noindent In Algorithm~\ref{alg:frank_wolfe}, each player can access an \emph{oracle} \change{that} returns the cost for a given joint state-action distribution. 
\change{In line~\ref{alg:mdp}, $\pi^i \in [A]^{(T+1)S}$ is any deterministic policy that solves the finite time MDP with cost $C^{ik}$, transition probability $P^i$, and initial distribution $z^i_0$. We use value iteration to recursively find $\pi^i$ as}
\change{\begin{equation}
\begin{aligned}\label{eqn:value_iteration}
    V^i_{Ts} & = \textstyle\min_{a}C^{ik}_{Tsa}, \ \pi^i_{Ts} \in  \argmin_{a}C^{ik}_{Tsa}, \\
    V^i_{(t-1)s} &  \textstyle = \min_{a} C^{ik}_{(t-1)sa} + {\sum_{s'}}P^i_{ts'sa}V^i_{ts'} \ \forall t \in [T] \\
    \pi^i_{(t-1)s} & \textstyle \in \argmin_{a} C^{ik}_{(t-1)sa} + {\sum_{s'}}P^i_{ts'sa}V^i_{ts'} \ \forall t \in [T]
\end{aligned}
\end{equation}}
\noindent \change{Algorithm~\ref{alg:frank_wolfe}} then retrieves the corresponding state-action density $b^{ik}$ via Algorithm~\ref{alg:density} and \change{combines it} with the current state-action density $x^{ik}$ to derive the next joint state-action density. All steps within lines 4 to 7 are parallelizable.

\begin{algorithm}[ht!]
\caption{Retrieving \change{state-action distribution} from \change{$\pi$} }
\begin{algorithmic}[1]
\Require \(P\), \(z\), \(\pi\).
\Ensure \(\{d_{tsa}\}_{t \in\mc{T}, s\in[S], a\in[A]}\)
\State{\(d_{0s\pi_{0s}} = z_s, \ \forall s \in [S]\)}
\For{\(t=1, \ldots, T\)}
		    \State{\(d_{ts(\pi_{ts})} = \textstyle\sum_a\sum_{s'} P_{tss'a}d_{(t-1)s'a}, \ \forall\ s \in [S]\)}
\EndFor 
\end{algorithmic}
\label{alg:density}
\end{algorithm}

\begin{theorem}
Under Assumption~\ref{assum:positive_definite}, 
Algorithm~\ref{alg:frank_wolfe} converges towards the Nash equilibrium \change{$\hat{x} = (\hat{x}^1,\ldots, \hat{x}^N)$} as 
\begin{equation}\label{eqn:alg_convergence}
   \textstyle \frac{\alpha}{2}\sum_{i\in[N]} \norm{x^{ik} - \change{\hat{x}^{i}}}^2_2 \leq \frac{2C_F}{k+2} 
\end{equation}
where $C_F$ is the potential function $F$'s~\eqref{eqn:potential_first_order} \emph{curvature constant} given by
\[ C_F = \underset{\substack{x^i, s^i \in \mc{X}(P^i, z^i_0) \\  \gamma \in [0, 1]\\
w^i  = x^i + \gamma(s^i - x^i)}}{\sup} \frac{2}{\gamma^2} \Big(F(\change{s}) - F(x) - \sum_{i \in [N]}(x^i - w^i)^\top \ell^i(x)\change{\Big)}.\]
\end{theorem}
\begin{proof}
Algorithm~\ref{alg:frank_wolfe} is a straight-forward implementation of~\cite[Alg.2]{jaggi2013revisiting}. 
    \change{From Assumption~\ref{assum:positive_definite},  $\nabla \xi(\hat{x}) \succ 0$. Therefore}, the potential function $F$ is \change{strongly} convex and satisfies $ \frac{\alpha}{2}\sum_{i\in[N]} \norm{x^{ik} - \change{\hat{x}^{i}}}^2_2 \leq F(x^{k}) - F(\hat{x})$. Equation~\eqref{eqn:alg_convergence} then follows directly from~\cite[Thm.1]{jaggi2013revisiting}.
\end{proof}
\begin{remark}[Scalability]
Algorithm~\ref{alg:frank_wolfe} has linear complexity in the number of players.
\end{remark}
\section{Multi-agent path coordination}\label{sec:mapf}
\noindent We apply our game model to a multi-agent pick up and delivery scenario \change{with stochastic package arrival times}. As shown in Figure~\ref{fig:grid_world}, $N$ players navigate a 2D \change{space}. \change{Each player's goal is} to  transport packages from the pick up chutes to the drop off chutes while \change{avoiding collision with others}. Code for the simulation is available at \url{https://github.com/lisarah/mdp_path_coordination}.
\begin{figure}
    \centering
    \vspace*{0.3cm} 
    \includegraphics[width=0.7\columnwidth]{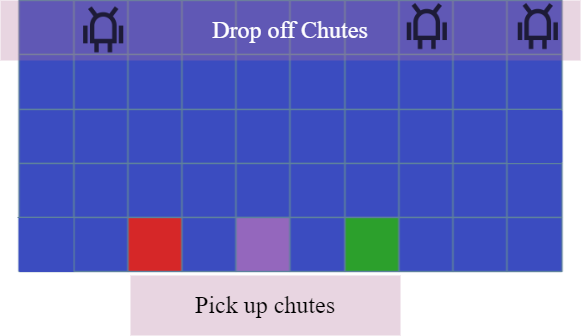}
    \caption{Operation environment for multi-robot warehouse scenario.}
    \label{fig:grid_world}
\end{figure}
\subsection{Stationary MDP Model}
\change{Players operate in a} two dimensional grid world with $5$ rows and $10$ columns. In addition to capturing location, each state also dictates whether the robot is in pick up or delivery mode. The state space is given by
\[[S] = \Big\{ (\change{v, w}, m) \ | \ 1\leq \change{v} \leq 5, \ 1\leq \change{w} \leq 10, \ m \in \change{\{1,2\}}\Big\}.\]
At each state, available actions are $[A] = \{u, d, r, l, s\}$, corresponding to up, down, right, left, stay. Player transition dynamics and rewards are \emph{stationary} in time.
The transition probability of each state $(\change{v, w}, m)$ extends the location-based transition probabilities $P^0$. 

\noindent\textbf{Location-based transition}. Let $u =\change{(v, w)}$ denote the location component of the state. At each location, each action either points to a feasible target $u_{targ}(a)$ or is infeasible. The set of all feasible targets from $u$ is $\mc{N}(u)$. 
When a target exists, players have $1 > q > 0$ chance of reaching it and $1 - q$ chance of reaching other states in $\mc{N}(u)$. 
\begin{equation}\label{eqn:feasible_location_transitions}
P^0_{u'ua} = \begin{cases}
q & u' = u_{targ}(a), \\
\frac{1 - q}{|\mathcal{N}(u)|} & u' \in \mathcal{N}(u)/\{u_{targ}(a)\},\\
0 & \text{otherwise}.
\end{cases}
\end{equation}
When the target location is infeasible, the player transitions into a neighboring state $u' \in \mc{N}(u)$ at random.
\begin{equation}\label{eqn:infeasible_location_transitions}
P^0_{u'ua} = \begin{cases}
\frac{1}{|\mathcal{N}(u)|} & u' \in \mathcal{N}(u),\\
0 & \text{otherwise}.
\end{cases}
\end{equation}
\textbf{Full transition dynamics}. \change{Within the same mode}, players transition between locations via dynamics $P^0$. Player modes \change{transition at pick up chutes $\mc{P}$ and drop off chutes $\mc{D}$}.
\change{
\begin{enumerate}
    \item When player $i$ is in mode $1$ (pick up) and about to transition into pick chute $p^i \in \mc{P}$, player $i$'s mode has $r^i$ probability of switching to mode $2$ (drop off). 
    \[\begin{cases}
     P^i_{t(p^i, 2)sa} = r^iP^0_{tp^iua},\\
     P^i_{t(p^i, 1)sa} = (1 - r^i)P^0_{tp^iua},
    \end{cases}\  \forall s = (u, 1), \ s \in [S].\]
    \item When player $i$ is in mode $2$ (drop off) and about to transition into drop chute $d^i \in \mc{D}$, player $i$ transitions to mode $1$ with probability $1$.
    \[\textstyle \begin{cases}
     \textstyle P^i_{t(d^i, 1)sa} = P^0_{td^iua},\\
     \textstyle P^i_{t(p^i, 2)sa} = 0,
    \end{cases}\  \forall s = (u, 2), \ s \in [S].\]
\end{enumerate}
}
Here, $r^i \in \reals$ denotes the probability of package arrival when player $i$ is in \change{$p^i$}. Modeled as an independent Poisson process with rate $\lambda_i$ and interval $\Delta t = 1s$, $r^i = \exp(-\lambda_i \Delta_t)$.

\subsection{Player Costs}
\noindent For all $(t,s, a) \in \mc{T}\times[S]\times[A]$ and congestion distribution $y$~\eqref{eqn:congestion_distribution}, player $i$'s \change{cost is given by
\[\textstyle\ell_{tsa}^i(y, x^i) = \epsilon x^i_{tsa} - c^i_{tsa} + \alpha_i f_{ts}(y).\]}
\noindent The player-specific objective $c^i_{tsa}$ is defined as 
\begin{equation}
c^i_{t(\change{v, w}, m) a} = \begin{cases}
1 & \change{(v, w)  = p^i}, \ m = \change{1}, \\
1 & \change{(v, w)  = d^i}, \ m = \change{2}, \\
0 & \text{otherwise.}
\end{cases}
\end{equation}
The congestion function is strictly state-based and is an exponential function given by
\begin{equation}\label{eqn:simulation_congestion_function}
f_{t\change{(v, w, m)}}(y) = -\beta\exp\big(\beta(\sum_{m' \in \change{\{1,2\}}}\sum_{a' \in [A]}y_{t(\change{v, w}, m')a'} - 1)\big),
\end{equation}
where $\alpha_i> 0$ for all $(t,s,a) \in \mc{T}\times[S]\times[A]$. \change{As opposed to~\eqref{eqn:individual_costs}, function~\eqref{eqn:simulation_congestion_function} calculates the congestion in $(v,w, \cdot)$ using both $(v,w, 1)$'s and $(v, w, 2)$'s congestion level. }
\begin{figure}
    \centering
    \vspace*{0.3cm} 
    \includegraphics[trim={0 0 0 0.5cm},width=0.85\columnwidth]{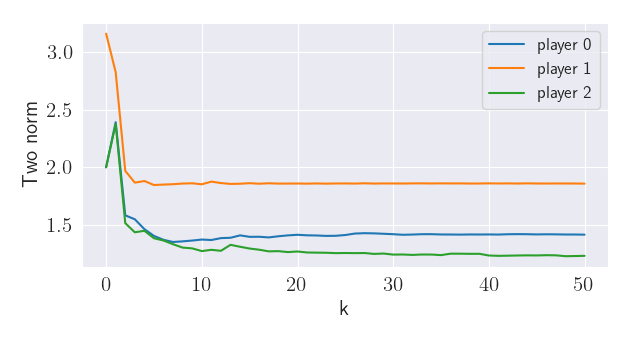}
    \caption{$\norm{\cdot}_2$ of player $i$'s state-action distribution over Algorithm~\ref{alg:frank_wolfe} iterations.}
    \label{fig:frank_wolfe}
\end{figure}
\subsection{Simulation Results}
\noindent We simulate the path coordination game using parameters from Table~\ref{tab:sim_params}. Player $i$'s pick up locations is the $i^{th}$ element of $ \mc{P}= \{(4, \change{w^i}) \ | \ \change{w^i} \in [8, 7, 2]\}$, \change{and its drop-off location is the $i^{th}$ element of $\mc{D} = \{(0, w^i) \ | \ w^i \in [4,5,8]\}$}. At $t = 0$, players are initialized \change{at their drop off location. }
\begin{table}[h!]
\begin{center}
\begin{tabular}{|ccccccccc|}
\hline
$N$  & $q$ & $\gamma_i$ & $\lambda_i$ & \change{$\alpha_i$} &  $\Delta t$ & $T$ & \change{$\epsilon$} & \change{$\beta$} \\
\hline
3 & 0.98 & 0.99& 0.5 & \change{\{0.5, 1, 1.5\}} & 1s & \change{120s} & \change{1e-3} & \change{40}\\
\hline
\end{tabular}
\end{center}
\caption{Parameters for simulation environment.}\label{tab:sim_params}
\end{table}

\noindent We run Algorithm~\ref{alg:frank_wolfe} for $100$ iterations, where line~\ref{alg:mdp} is solved via value iteration~\eqref{eqn:value_iteration}. \change{The two norm of $x^i$ is shown in Figure~\ref{fig:frank_wolfe} as a function of the algorithm iterations. We see that the state-action densities stabilize in about $20$ steps.} Performance is evaluated by: 1) expected number of collisions, 2) expected packages delivery time, 3) worst package delivery time. The results over $100$ random trials are visualized in Figures~\ref{fig:collision_results} and~\ref{fig:waiting_time_results}. 
\begin{figure}
    \centering
    \vspace*{0.3cm} 
    \includegraphics[trim={0 0 0 0.5cm},width=0.8\columnwidth]{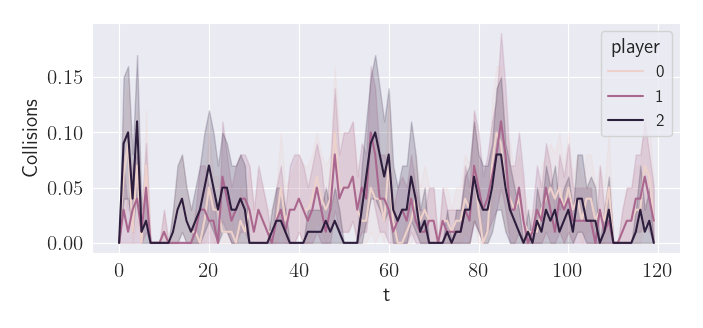}
    \caption{Collisions per player as a function of MDP time step $t$.}
    \label{fig:collision_results}
\end{figure}
\begin{figure}
    \centering
    \includegraphics[width=0.95\columnwidth]{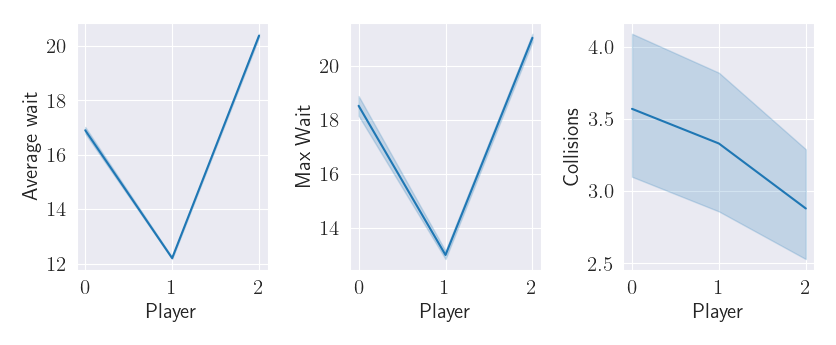}
    \caption{Average waiting time per package, worst case waiting time per package, and average number of collisions in $T$ for each player.}
    \label{fig:waiting_time_results}
\end{figure}

\noindent\change{
We compare the \emph{jointly optimal congestion-free wait time} computed using Algorithm~\ref{alg:frank_wolfe}, and compare them to the shortest \change{wait time} available in the absence of opponents.
}
Each path is the number of steps to complete the drop off-pick up-drop off cycle. \change{Based on players' pick-up and drop-off locations, their shortest \change{wait time} in the absence of opponents is $16$, $12$, $20$ respectively.} We note that this matches well with the average wait time shown in Figure~\ref{fig:waiting_time_results}.

We set the player impact factors as $\{0.5, 1, 1.5\}$ as in Table~\ref{tab:sim_params}. From Figure~\ref{fig:waiting_time_results}, the impact factors directly correlate with rate of collision players experience. Player $0$ impacts congestion the least and is the least sensitive to congestion. As a result, it encountered more collisions. Player $2$ impacts congestion the most and is the most sensitive to congestion. As a result, it encountered the least collisions. The collision rate is spread out evenly over $\mc{T}$ (Figure~\ref{fig:collision_results}).
\section{Conclusions}\label{conclusion}
We derived a class of $N$ player, weighted potential games under heterogeneous MDP dynamics and with application to multi-agent path coordination. For these games, we show equivalence between the unique Nash equilibrium and the global solution of a potential minimization problem, which we solve via gradient descent and single-player dynamic programming. Future work include deriving learning-based solutions for the games and \change{integrating partially observable scenarios in which players have local observations only.}

\bibliographystyle{IEEEtran}
\bibliography{reference}
\appendix
\begin{proposition}\label{prop:dp_to_minimizer} 
Under Assumption~\ref{assum:positive_definite}, consider the problem
\begin{equation}\label{eqn:mdp_variant}
    \min_{x^i} F(x^i, x^{-i}) \, \text{ s.t. } x^i \in \mc{X}(P^i, z_0^i).
\end{equation} 
where for $i \in [N]$ and $x^{-i}$, the objective $F:\reals^{N(T+1)SA}\mapsto\reals$ satisfies $\partial F(x^i, x^{-i})/\partial x^i= \ell^i(x^i, x^{-i})$ $\forall x^i \in \mc{X}(P^i, z_0^i)$, then $\hat{x}^{i}$ minimizes~\eqref{eqn:mdp_variant} if and only if  $Q^i(\hat{x}^{i}, x^{-i})$ in~\eqref{eqn:q_value} satisfies~\eqref{eqn:individual_optimality}. 
\end{proposition}
\begin{proof}
Because~\eqref{eqn:mdp_variant} has linear constraints and $\partial^2F(x)/\partial x_i^2$ $= \partial \ell^i(x)/\partial x_i\succ 0$ by assumption, \eqref{eqn:mdp_variant}'s unique minimizer satisfies the first order KKT conditions. 
\noindent Consider the dual variables $\mu^i $ $\in \reals^{(T+1)SA}_+$ for $x^i\geq 0$ and $\nu^i \in$ $ \reals^{(T+1)SA}$  for the equality constraints in $\mc{X}(P^i, z_0^i)$~\eqref{eqn:feasible_mdp_flows}. The Lagrangian of~\eqref{eqn:mdp_variant} is $L(x^i, \nu^i, \mu^i)  = F(x^i, x^{-i}) - \sum_{t,s,a} {\mu}^i_{tsa} x^i_{tsa} + \sum_{s} \nu^i_{0s}(x^i_{0s} - \sum_{a} x^i_{0sa}) + \sum_{s, t} \nu^i_{ts} (\sum_{s'a} P^i_{tss'a}x^i_{(t-1)sa} - \sum_{a}x^i_{tsa})$.
The KKT conditions are 1) primal feasibility $x^i \in \mc{X}(P^i, z^i_0)$, 2) dual feasibility $\mu^i \geq 0$, 3) complementary slackness $\mu^i_{tsa}x^i_{tsa} = 0, \ \forall (t, s, a) \in \mc{T}\times[S]\times[A]$, and 4) stationarity condition, given $\forall (t,s,a) \in \mc{T}\times[S]\times[A]$ as
\begin{equation}\label{eqn:og_kkt}
\begin{cases}
\ell^i_{tsa}(x) + \sum_{s'} P^i_{(t+1)s'sa}\nu^i_{(t+1)s'}  =  \nu^i_{ts}  + {\mu}^i_{tsa}& t \neq T \\
\ell^i_{Tsa}(x)  =  \nu^i_{Ts} + {\mu}^i_{Tsa}  & t = T.
\end{cases}
\end{equation}
We can show that $(\hat{x}^{i}, x^{-i})$ satisfies the KKT conditions above if and only if it satisfies~\eqref{eqn:individual_optimality}. To simplify notation, we use $Q^i_{tsa}$ to denote $Q^i_{tsa}(\hat{x}^{i}, x^{-i})$. 

\noindent ($\Rightarrow$): suppose $(\hat{x}^{i}, \nu^i, \mu^i)$ satisfies the KKT conditions. When $\hat{x}^i_{tsa}>0$, $\nu^i_{ts}$ represents the value function and $\nu^i_{ts} + \mu^i_{tsa}$ represents $Q$-value. When $\hat{x}^i_{tsa}=0$, we \emph{shift} $(\nu^i, \mu^i)$ to $(\hat{\nu}^i, \hat{\mu}^i)$ to generate the optimal Q-values. To this end, define $\lambda^i \in\reals^{(T+1)SA}$, $\Delta^i\in\reals^{(T+1)S}, \hat{\mu}^i \in \reals^{(T+1)SA}, \hat{\nu}^i\in \reals^{(T+1)S}$ recursively from $t = T$.  At $T=t$, let $\Delta^i_{(T+1)s'} = 0 \ \forall s' \in[S]$. All other variables are recursively defined as
\begin{equation}\label{eqn:pseudo_q_values}
    \begin{aligned}
        \lambda^i_{tsa}& \textstyle= \mu^i_{tsa} + \sum_{s'}P^i_{(t+1)s'sa}\Delta^i_{(t+1)s'},  \\
        \textstyle\Delta^i_{ts} &=\textstyle \min_{a'} \lambda^i_{tsa'}, \\
        \textstyle\hat{\mu}^i_{tsa} &= \textstyle\lambda^i_{tsa} - \Delta^i_{ts},\\
        \textstyle\hat{\nu}^i_{ts} &= \textstyle{\nu}^i_{ts} + \Delta^i_{ts}.
    \end{aligned}
\end{equation}
At time $t$, let the condition $\hat{x}^i_{tsa} > 0$ implies $\lambda^i_{tsa} = 0$ be denoted as $\mc{K}(t)$. We can show that $\mc{K}(t)$ implies $\mc{K}(t-1)$: from complementary slackness, $\hat{x}^i_{(t-1)sa} > 0$ implies $\mu^i_{(t-1)sa} = 0$. Subsequently, $\lambda^i_{(t-1)sa} = 0$~\eqref{eqn:pseudo_q_values} if $P^i_{ts'sa}\Delta^i_{ts'} = 0 \ \forall s' \in [S]$: either $P^i_{ts'sa}  = 0$ or $P^i_{ts'sa}\hat{x}^i_{(t-1)sa} = \sum_{a'}\hat{x}^i_{ts'a'}  > 0$. In the second case, there exists $a'\in [A]$ such that $\hat{x}^i_{ts'a'} > 0$, and if $\mc{K}(t)$ holds, $\lambda^i_{ts'a'} = 0$. By definition, $\Delta^i_{ts'}$ is non-negative and must be zero. We conclude that $P^i_{ts'sa}\Delta^i_{ts'} = 0$ $\forall s' \in [S]$, and $\mc{K}(t-1)$ holds. At $t=T$, $\hat{x}^i_{Tsa} > 0$ implies $\mu^i_{Tsa} = 0$ and $\lambda^i_{Tsa}= 0$. Therefore, $\mc{K}(t)$ holds $\forall t \in \mc{T}$.  

\noindent By adding  $\sum_{s'}P^i_{(t+1)s'sa}\Delta^i_{(t+1)s'}$ to~\eqref{eqn:og_kkt} and simplifying it via~\eqref{eqn:pseudo_q_values}, we obtain 
\begin{equation}\label{eqn:thm1_pf4}
\begin{cases}
\ell^i_{tsa}(x) + \sum_{s'} P^i_{(t+1)s'sa}\hat{\nu}^i_{(t+1)s} =  \hat{\nu}^i_{ts}  + \hat{\mu}^i_{tsa} & t \neq T \\
\ell^i_{Tsa}(x)  =  \hat{\nu}^i_{Ts} + \hat{\mu}^i_{Tsa} & t = T.
\end{cases}
\end{equation}
We define $Q^i_{tsa} = \hat{\nu}^i_{ts} + \hat{\mu}^i_{tsa}$. From~\eqref{eqn:pseudo_q_values}, $\hat{\mu}^i_{tsa}$ is always non-negative and $\hat{\mu}^i_{tsa'} = 0$ for some $a'\in[A]$. Therefore $\min_{a'}Q^i_{tsa'} = \hat{\nu}^i_{ts}$, and $Q^i$ substituted in~\eqref{eqn:thm1_pf4} satisfies~\eqref{eqn:q_value}. 

\noindent  If $\hat{x}^{i}_{tsa}> 0$, then from $\mc{K}(t)$, $\lambda^i_{tsa} =0$. Therefore, $\hat{\mu}^i_{tsa} = 0$ and $Q^i_{tsa} = \min_{a'}Q^i_{tsa'}$. We conclude that $Q^i$ satisfies~\eqref{eqn:individual_optimality}.

\noindent($\Leftarrow$): We show that if $Q^i$ satisfies~\eqref{eqn:individual_optimality}, then $\hat{x}^i$ satisfies the KKT conditions. Let $\nu^i_{ts} = \min_{a'}Q^i_{tsa'}$ and $\mu^i_{tsa} = Q^i_{tsa} - \nu^i_{ts}$ $\forall (t, s, a) \in \mc{T}\times[S]\times[A]$, then $(\hat{x}^{i}, \nu^i, \mu^i)$ is a KKT point.
Both $\hat{x}^{i}$ and $\mu^i$ satisfy primal/dual feasibility respectively. From~\eqref{eqn:individual_optimality}, $\hat{x}^{i}_{tsa} > 0$ implies that $\nu^i_{ts} = Q^i_{tsa}$ and $\mu^i_{tsa} = 0$. Since either $\hat{x}^{i}_{tsa} > 0$ or $\hat{x}^{i}_{tsa} = 0 $, complementary slackness $\hat{x}^{i}_{tsa} \mu^i_{tsa} = 0$ holds $\forall (t, s, a)\in \mc{T}\times[S]\times[A]$. Finally, the stationarity condition~\eqref{eqn:og_kkt} directly follows from~\eqref{eqn:q_value}. 
\end{proof}

\end{document}